\documentclass[10pt]{article}
\usepackage{amssymb}
\usepackage{amsmath}
\usepackage[all]{xy}
\usepackage{amsmath}
\usepackage{amsfonts}
\usepackage{amssymb}
\usepackage{amscd}
\usepackage{amsthm}
\usepackage{latexsym}
\usepackage{amsbsy}
\usepackage{graphicx}
\usepackage{float}
\usepackage[toc,page]{appendix}
\usepackage{caption}
\usepackage{subcaption}

\usepackage{amssymb,latexsym,amsmath,amscd,slashed,mathtools,mathrsfs,bm,stmaryrd,dsfont,amsthm}
\usepackage{setspace,scalerel,color,graphicx,enumerate}

\newcommand{\ci}{\mathrm{i}}

\newtheorem*{remark}{Remark}
\newtheorem{theorem}{Theorem}[section]
\newtheorem{definition}{Definition}[section]
\newtheorem{corollary}{Corollary}[theorem]
\newtheorem{lemma}[theorem]{Lemma}
\newtheorem{prop}{Proposition}
\DeclareMathOperator{\tr}{Tr}

\makeatletter
\def\section{\@startsection{section}{1}{\z@}{-3.25ex plus -1ex minus
		-.2ex}{1.5ex plus .2ex}{\normalfont\bfseries}}
\def\subsection{\@startsection{subsection}{1}{\z@}{-3.25ex plus -1ex
		minus -.2ex}{1.5ex plus .2ex}{\normalfont\itshape}}
\date{Resubmission date: April 4, 2022}
\makeatother
\title{ Spectral statistics of Dirac ensembles}
\author{ Masoud Khalkhali and  Nathan Pagliaroli\\
	Department of Mathematics, University of Western Ontario\\
	London, Ontario, Canada\footnote{\emph{Email addresses}:  masoud@uwo.ca, npagliar@uwo.ca}}

\begin{document}
	\maketitle
	\begin{abstract}
	In this paper we find spectral properties in the large $N$ limit of Dirac operators that come from random finite noncommutative geometries. In particular for a Gaussian potential the limiting eigenvalue spectrum is shown to be universal regardless of the geometry and is given by the convolution of the semicircle law with itself. For simple non-Gaussian models this convolution property is also evident. In order to prove these results we show that a wide class of multi-trace multimatrix models have a genus expansion.  
	\end{abstract}

\section{Introduction}	

The notion of a {\it Dirac ensemble} provides an interesting link between noncommutative geometry and random matrix theory. The {\it partition function} of these ensembles is   of the form 
\begin{equation}\label{partition}
Z = \int e^{- \tr S (D)}  d D \, ,
\end{equation}
where the potential functional ${S (D)}$ is defined in terms of the spectrum of the Dirac operator $D$, and the integral is over the moduli space of Dirac operators compatible with a  fixed  {\it finite noncommutative geometry}, called the {\it Fermion space}.  In particular Dirac operators are dynamical variables and play the role of metric  fields in these models.  Moreover the moduli space of Dirac operators is typically a finite dimensional vector space. The link to random matrix theory is through the  associated multimatrix  and multi-trace random  matrix integral of the form
$$ Z =\int e^{- V(H_1, H_2, \dots, H_k)} dH_1 dH_2 \dots dH_k,$$ 
where the potential $V$   is derived from the potential functional    $S$ in equation (\ref{partition}). More generally one is interested in expectation values of the form
$$ \langle \, \mathcal{ O}   \rangle  = \int \mathcal{ O}(H_1, \dots, H_k) e^{- V(H_1, H_2, \dots, H_k)} dH_1 dH_2 \dots dH_k,$$
where $\mathcal{O}$, the observable,  is a  polynomial function in matrix variables  and  integration is over the space of $k$-tuples of Hermitian matrices  with its  Lebesgue measure. 
It should be stressed that the potential function $S$  in \eqref{partition} is usually chosen in such a way that the partition function  is absolutely convergent and finite. However, divergent integrals can be studied perturbatively as formal matrix integrals, which we will briefly discuss in Appendix B. A typical  choice for $S$ would be 
$$S(D)= \text{Tr} (f (D))$$ for a  real polynomial $f$ of even degree with positive leading coefficient. This is in contrast with the spectral action principle of Chamseddine and Connes, where the heat kernel expansion of $f(D)$ for a rapid decay even function $f$ plays a dominant role \cite{Chamseddine}. 

In this way techniques of random matrix theory  such as 't Hooft  genus expansion, resolvent methods, Schwinger-Dyson equations,   spectral curves, and topological recursion provide immediate and very natural links between noncommutative geometry, classical geometry, and analysis on  Riemann surfaces. This idea of using random matrix theory techniques to study Dirac ensembles like  (\ref{partition})  was first pursued in \cite{AK,First paper} and the present  paper should be regarded as a contribution to this idea. Another recent idea was to employ Bootstrapping to these models \cite{Bootstraps}. The use of random matrix theory techniques can also be found in Noncommutative Quantum Field Theory \cite{Grosse all quartic,Blobbed top of quartic}. 

An alternative method of studying Dirac ensembles  would be through use of Monte Carlo simulation. This is the approached pursued by Barret and Glaser in  \cite{Barrett2016} where these models were first introduced. This was further explored in \cite{glaser,Spectral estimators}. The motivation was to give toy models of Euclidean quantum gravity. We should also mention that in \cite{BV nottingham} BV formalism is applied to analyze these models.

In this paper we explicitly find the eigenvalue distribution for all Gaussian Dirac ensembles, i.e. ensembles of the form 
\begin{equation*}
Z = \int_{\mathcal{G}}e^{-\frac{1}{2k}\tr D^{2}}dD,
\end{equation*}
where $\mathcal{G}$ is the moduli space of Dirac operators and $k$ is some appropriate integer for normalization that depends on the spectral triple. In fact, $k$ is the dimension of the space of gamma matrices from the fermion space. We refer to this result as the Wigner Convolution law  and it goes as follows: for any Gaussian Dirac ensemble, the limiting spectral density function of the Dirac operator is given by
\begin{equation}\label{convolution}
\rho_{D}(x) = \int_{\mathbb{R}}\rho_{W}(x-t)\rho_{W}(t)dt,
\end{equation}
where
\begin{equation*}
\rho_{W}(x) = \frac{1}{2\pi}\sqrt{4-x^{2}}_{[-2,\,2]}, 
\end{equation*}
is the Wigner Semicircle Distribution.
This result is interesting because it is independent of the geometry of the Dirac ensemble. The integral (\ref{convolution}) is elliptic and does not have a closed form.  Observe  figure 1 for a comparison of the semicircle distribution with its self-convolution. This result is proved in Section four.
\begin{figure}[H]\label{Wigner}
	\includegraphics[width=8cm]{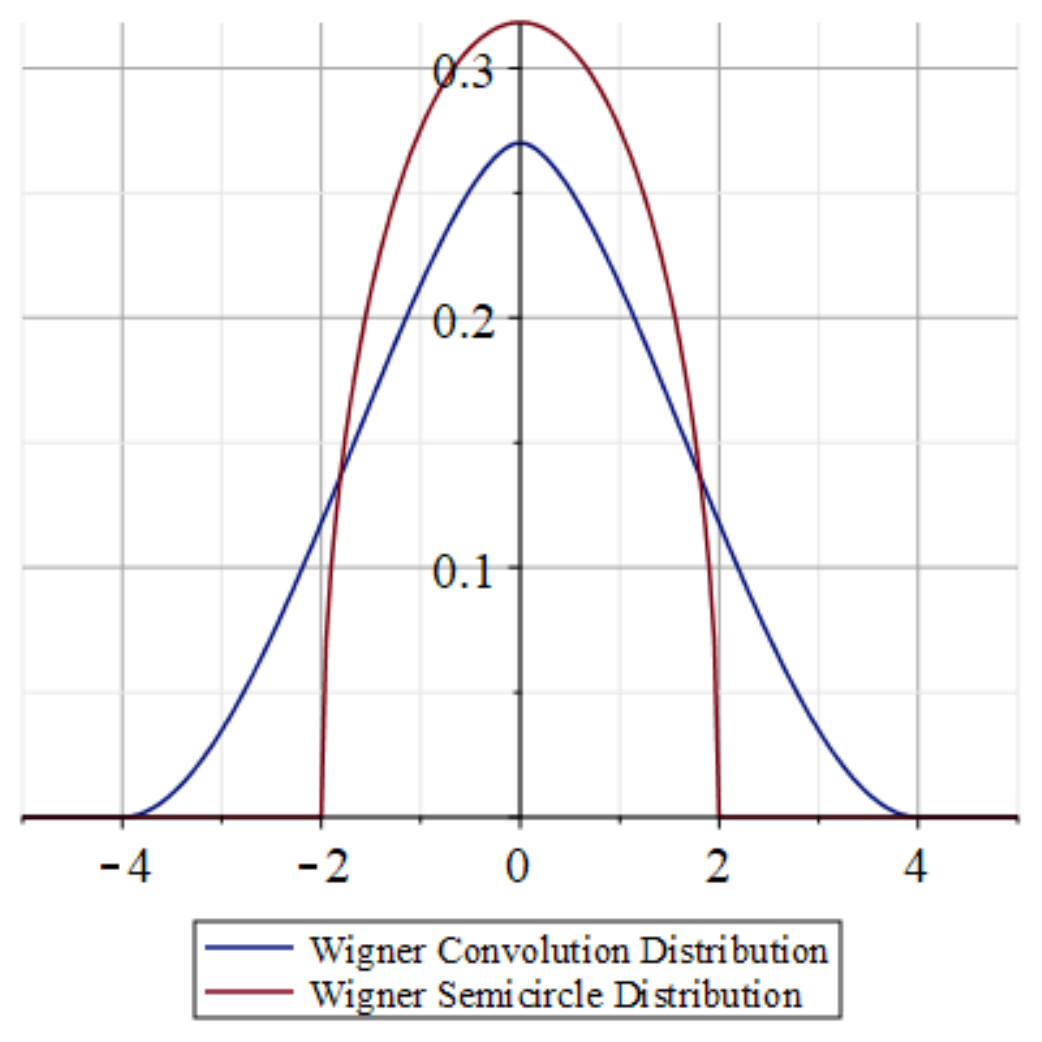}
	\centering
	\caption{The Wigner semicircle distribution compared to the Wigner  Convolution Distribution.}
\end{figure}

A second related result we prove in this paper is that for Dirac ensembles that are single matrix models a similar law holds.  Such Dirac ensembles are necessarily of type $(1,0)$ and $(0,1)$ \cite{Barrett2015}. Let $\rho$ be the limiting eigenvalue distribution for the corresponding random matrix model. Then the limiting eigenvalue distribution of the Dirac ensemble is
\begin{equation*}
\rho_{D}(x) = \int_{\mathbb{R}}\rho(x-t)\rho(x)dt.
\end{equation*}

It turns out that in the Gaussian case $\rho $ is the semicircle distribution, thus both results tell us that the Dirac operator's spectral density function is the self-convolution of the random matrix model's spectral density function. It would be interesting to see if similar results apply to an even wider class of models, but scarcity of techniques concerning multimatrix models  makes this task difficult.

This paper is organized as follows. In Section two we define precisely what we mean by a Dirac ensemble and introduce some basic examples. Each Dirac ensemble has a corresponding random matrix model. We then discuss the general relationship between the Dirac operators spectral density function and the corresponding random matrix spectral density function in the large $N$ limit. In Section three we review the relationship between stuffed maps and bitracial matrix models, as seen in \cite{blobbed1,blobbed,blobbed2}. We provide examples of stuffed maps glued from various 2-cells and consider the generating functions of such gluings. These generating functions are proven to be well-defined objects with a genus expansion. Furthermore they satisfy a 1-cut Lemma which provides a factorization of the resolvent type generating function. In Section four we prove the main convolution theorems mentioned above. In Appendix A, we explain why certain terms in multitrace matrix models do not contribute to the large $N$ limiting eigenvalue distribution. In Appendix B, we review the definition of formal matrix models.

\section{Examples of Dirac ensembles}

It was discovered via the Connes' distance formula \cite{Connes94}
$$ d(p, q) = \text{Sup}\{|f(p)-f(q)|; \, ||[D, f]|| \leq 1\},$$
that the geometric distance on a spin manifold can be recovered from the Dirac operator $D$ on the space of spinors. The  {\it reconstruction theorem} of Connes  \cite{Connes-2013},   tells us that   a spin Riemannian manifold can be recovered from  a commutative real spectral triple under certain additional conditions. Thus we may think of real  spectral triples  as a noncommutative analogue of spin Riemannian manifolds, where the Dirac operator defines the metric.

A spectral triple has three main components $(\mathcal{A},\mathcal{H}, D)$, where $\mathcal{A}$ is an involutive complex algebra acting by bounded operators on a Hilbert space $\mathcal{H}$ and  $D$ is a  self-adjoint operator acting on $\mathcal{H}$ \cite{Connes NCG and Reality}.  A real spectral triple additionally comes with the charge conjugation operator  $J$ and  the chirality operator $\gamma$.  Finite dimensional real spectral triples have been classified in \cite{Barrett2015}. In particular, the number of gamma matrices that square to one and minus one, denoted $(p,q)$ respectively,  can be used to characterize finite real spectral triples. In this paper  we are strictly interested in finite dimensional real spectral triples which allows our integrals to be expressed as matrix integrals.

Let us denote by ${{\mathrm{C \ell}}_{p,q}}$  the real Clifford algebra of the real quadratic space  ${\mathbb{R}^{p,q}}$ equipped with the quadratic form
\begin{equation*} 
{\eta} (v,v) = {{v_1}^2} + \cdots + {{v_p}^2} - {{v_{p+1}}^2} - \cdots -  {{v_{p+q}}^2} \,
, \quad v \in \mathbb{R}^{p,q} \, .
\end{equation*}
Let ${ {\mathbb{C} \ell}_n = {\mathrm{C \ell}}_{p,q} {\otimes_{\mathbb{R}}} \mathbb{C}\,}$, with $p+q=n$, denote the  complexification of  ${{\mathrm{C \ell}}_{p,q}}$. Let ${ {\{ e_i \}}_{i=1}^n }$ be an oriented  basis of $\mathbb{R}^{p,q}$ with ${{\eta} (e_i , e_j) = \pm \delta_{ij} \,}$. The \emph{chirality operator} $\Gamma$ is defined by
\begin{equation*}
\Gamma = \ci^{\frac{1}{2} s (s+1)} \, e_1 e_2 \cdots e_n \, ,
\end{equation*}
with $s \equiv q-p \pmod{8}$.
In this paper, ${V_{p,q}}$ denotes the irreducible complex ${{\mathrm{C \ell}}_{p,q} \,}$-module, where, for $p+q =n$ odd, the chirality operator $\Gamma$ acts trivially on ${V_{p,q} \,}$. The operators $\gamma^{i} = \rho(e_{i})$ are called \textit{gamma matrices}. It is well known that there exist a Hermitian inner product ${\langle \cdot , \cdot \rangle}$ on ${V_{p,q}}$ such that the gamma matrices act as unitary operators.

Let ${C : V_{p,q} \to V_{p,q}}$ be a \emph{real structure} of ${KO}$-dimension 
${s \equiv q-p \pmod{8}}$ (see, e.g. \cite{Marcolli,Connes NCG and Reality}) on ${V_{p,q}}$ such that
\begin{equation*}
\left( { {\mathbb{C} \ell}_n \, , V_{p,q} \, , \Gamma , C} \right)
\end{equation*}
satisfies all the axioms of a fermion space, that is all the axioms of a real spectral triple except the existence of a Dirac operator. We borrow the following definition from \cite{AK}:

\begin{definition}
A \emph{matrix geometry of type ${(p,q)}$} is a finite real spectral triple
${\left( { \mathcal{A}, \mathcal{H} , D , \gamma , J} \right)}$, where the corresponding fermion space, that is given by:
\begin{itemize}
\item
${ \mathcal{A}= {{\mathrm{M}}_N (\mathbb{C})} }$
\item
${ \mathcal{H} = V_{p,q} \otimes {{\mathrm{M}}_N (\mathbb{C})} }$
\item
${ \langle {v \otimes A , u \otimes B} \rangle = \langle v , u \rangle \, \tr \left( A B^\ast \right)
\, , \quad v,u \in V_{p,q} \, , \  A,B \in {{\mathrm{M}}_N (\mathbb{C})} }$
\item
${ \pi (A) (v \otimes B) = v \otimes \left( AB \right)  }$
\item
${ \gamma (v \otimes A) = (\Gamma v) \otimes A }$
\item
${ J (v \otimes A) = (C v) \otimes A^\ast  \,}$,
\end{itemize}
where ${{{\mathrm{M}}_N (\mathbb{C})} \cong \mathrm{End} \left( {{\mathbb{C}}^N} \right) = 
{\mathbb{C}}^N \otimes \left( {{\mathbb{C}}^N} \right)^\ast \, }$. The Dirac operators of type $(p,q)$ matrix geometries are expressed in term of gamma matrices ${\gamma^i \,}$, and commutators or anti-commutators with given Hermitian matrices $H$ and skew-Hermitian matrices $L$ (see \cite{Barrett2015, Barrett2016}).\footnote{In \cite{Barrett2015}, this class of spectral triples is referred to as \emph{fuzzy spaces of type ${(p,q)}$}.}
\end{definition}

We now define a Dirac ensemble as a matrix geometry such that the Dirac operator is a random matrix distributed according to some matrix probability distribution 
\begin{equation*}
	e^{-\tr S(D)}dD
\end{equation*} 
while simultaneously satisfying the axioms of a real spectral triple. We now present to the reader some simple examples of such ensembles. 

\subsection{1-matrix Dirac ensembles}

Consider finite real spectral triples $(A, \mathcal{H}, D)$ where the algebra is $ A= M_N(\mathbb{C})$ and the Hilbert space is $ \mathcal{H} =\mathbb{C} \otimes M_N(\mathbb{C})$. The two noncommutative geometries with  $p+q =1$ from \cite{Barrett2016} are as follows:
\begin{enumerate}
	\item Type (1,0) with \begin{equation*}
	\gamma^{1} = 1,
	\end{equation*}
	\begin{equation*}
	D = \{H,\cdot\},
	\end{equation*}
	where $H$ is a Hermitian matrix.
	\item Type (0,1) with \begin{equation*}
	\gamma^{1} = -i,
	\end{equation*} 
	\begin{equation*}
	D = \gamma^{1}\otimes [L,\cdot],
	\end{equation*}
	where $L$ is a skew-Hermitian matrix.
\end{enumerate}
The commutator and anti-commutator can be written using the tensor product:
\begin{equation*}
	\{H,\cdot \} = H\otimes I_{N} + I_{N} \otimes H^{T},
\end{equation*} 
\begin{equation*}
	[L,\cdot ] = L\otimes I_{N} - I_{N} \otimes L^{T}.
\end{equation*} 
This allows us to compute trace powers of $D$ in both cases:
\begin{equation*}
	\sum_{k=0}^{\ell} {\ell\choose k}\tr H^{\ell-k}\tr H^{k},
\end{equation*}
\begin{equation*}
\sum_{k=0}^{\ell} {\ell\choose k}(-1)^{k}\tr L^{\ell-k}\tr L^{k}.
\end{equation*}

Consider for example the following quartic Dirac ensemble in both types

\begin{equation*}
	Z = \int_{\mathcal{G}} e^{-\frac{1}{4}\tr D^{2} - \frac{t_{4}}{8}\tr D^{4}} dD.
\end{equation*}
In type $(1,0)$ the integral is over the space of Hermitian $N\times N$ matrices and the potential is 
\begin{align*}
	\frac{1}{2}(N \tr H^{2} + 2(\tr H)^{2}) + \frac{1}{4}\left( N \tr H^{4} + 8 \tr H \tr H^{3} + 6 (\tr H^{2})^{2}\right).
\end{align*}
In type $(0,1)$ the integral is over the space of skew-Hermitian $N\times N$ matrices and the potential is 
\begin{align*}
\frac{1}{2}(-N \tr L^{2} + 2(\tr L)^{2}) + \frac{1}{4}\left( N \tr L^{4} -8 \tr L \tr L^{3} + 6 (\tr L^{2})^{2}\right).
\end{align*}
We may apply the transformation $L \rightarrow iH $, for some Hermitian matrix $H$, to get 
\begin{align*}
\frac{1}{2}(N \tr H^{2} - 2(\tr H)^{2}) + \frac{1}{4}\left( N \tr H^{4} - 8 \tr H \tr H^{3} + 6 (\tr H^{2})^{2}\right).
\end{align*}

As we will later see the two terms with minus signs contribute nothing in the large $N$ limit. This was first noticed for this type of convergent model in \cite{First paper}.  Using the above formulas for even trace powers it is not hard to see that type $(1,0)$ and $(0,1)$ will have identical limiting spectral density and moment generating functions. For further explanation we refer the reader to Appendix A. 

We wish to study that the limiting eigenavlue distribution of the Dirac operator using random matrix theory. The following theorem gives the relationship between the spectral density function of the Dirac operator to that of its random matrix model in the large $N$ limit when there is an even potential.

\begin{theorem}
	Consider a type (1,0) or (0,1) Dirac ensemble with a partition function
	\begin{equation*}
		Z = \int_{\mathcal{G}} e^{-\frac{1}{2k}\tr S(D)}dD
	\end{equation*}
	where
	\begin{equation*}
		S(D) = \frac{1}{2}D^{2} + \sum_{j=3}^{d}\frac{t_{2j}}{2j}D^{2j}.
	\end{equation*}
	If the limiting eigenvalue distributions of the associated random matrix ensemble exist (in the formal sense), call it $\rho(x)$, then
	 the limiting spectral density function of the Dirac operator is
	\begin{equation*}
	\rho_{D}(x) = \int_{\mathbb{R}}\rho(x-t)\rho(x)dt.
	\end{equation*}
\end{theorem}
The proof is presented in Section four. 
\begin{remark}
	It is often the case that one finds a convergent model's eigenvalue distribution coincides with its formal counterpart in the large $N$ limit. In such a case  Theorem 2.1 applies to the corresponding  convergent model. It will be discussed later once the 1-cut lemma is introduced as to when precisely this theorem applies to convergent models. The spectral density function for convergent matrix models of this type can be found using the methods in \cite{First paper}.
	
\end{remark}

\subsection{2-matrix Dirac ensembles} 
Consider finite real spectral triples $(A, \mathcal{H}, D)$ where the algebra is  $A= M_N(\mathbb{C})$ and the Hilbert space is $ \mathcal{H} = \mathbb{C}^{2} \otimes M_N(\mathbb{C})$. The three $p+q =2$ noncommutative geometries from \cite{Barrett2016} are as follows:
\begin{enumerate}
	\item Type (2, 0):  let  \begin{equation*}
	\gamma^{1} = \begin{pmatrix}
	1 & 0 \\
	0 & -1
	\end{pmatrix}, 
	\quad \quad \gamma^{2} = \begin{pmatrix}
	0 & 1 \\
	1 & 0
	\end{pmatrix}.
	\end{equation*}
	Then, 
	\begin{equation*}
	D = \gamma^{1}\otimes \{H_{1},\cdot\} + \gamma^{2} \otimes \{H_{2},\cdot\},
	\end{equation*}
	where $H_1$ and $H_2$ are Hermitian matrices. 
	
	\item Type (1,1):  let  \begin{equation*}
	\gamma^{1} = \begin{pmatrix}
	1 & 0 \\
	0 & -1
	\end{pmatrix},
	\quad \quad \gamma^{2} = \begin{pmatrix}
	0 & 1 \\
	-1 & 0
	\end{pmatrix}.
	\end{equation*}
	Then,  
	\begin{equation*}
	D = \gamma^{1}\otimes \{H,\cdot\} + \gamma^{2} \otimes [L,\cdot],
	\end{equation*}
	where $H$ is Hermitian and $L$ is skew-Hermitian.
	
	\item Type (0,2); let \begin{equation*}
	\gamma^{1} = \begin{pmatrix}
	i & 0 \\
	0 & -i
	\end{pmatrix},
	\quad \quad \gamma^{2} = \begin{pmatrix}
	0 & 1 \\
	-1 & 0
	\end{pmatrix}.
	\end{equation*}
	Then  
	\begin{equation*}
	D = \gamma^{1}\otimes [L_{1},\cdot] + \gamma^{2} \otimes [L_{2},\cdot],
	\end{equation*}
	where $L_1, L_2$ are both skew-Hermitian. 
\end{enumerate}
 Our goal will be to apply the substitution $L = iH$ for each skew-Hermitian matrix $L$ in the above geometries to get these geometries strictly in terms of Hermitian matrices. The transformed operators and gamma matrices are 
 \begin{enumerate}
 	\item Type (2,0) with \begin{equation*}
 	\gamma^{1} = \begin{pmatrix}
 	1 & 0 \\
 	0 & -1
 	\end{pmatrix}
 	\qquad \qquad \gamma^{2} = \begin{pmatrix}
 	0 & 1 \\
 	1 & 0
 	\end{pmatrix},
 	\end{equation*}
 	and 
 	\begin{equation*}
 	D = \gamma^{1}\otimes \{H_{1},\cdot\} + \gamma^{2} \otimes \{H_{2},\cdot\}.
 	\end{equation*}
 	\item Type (1,1) with \begin{equation*}
 	\gamma^{1} = \begin{pmatrix}
 	1 & 0 \\
 	0 & -1
 	\end{pmatrix}
 	\qquad \qquad \gamma^{2} = \begin{pmatrix}
 	0 & i \\
 	-i & 0
 	\end{pmatrix},
 	\end{equation*}
 	and 
 	\begin{equation*}
 	D = \gamma^{1}\otimes \{H_{1},\cdot\} + \gamma^{2} \otimes [H_{2},\cdot].
 	\end{equation*}
 	\item Type (0,2) with \begin{equation*}
 	\gamma^{1} = \begin{pmatrix}
 	-1 & 0 \\
 	0 & 1
 	\end{pmatrix}
 	\qquad \qquad \gamma^{2} = \begin{pmatrix}
 	0 & i \\
 	-i & 0
 	\end{pmatrix},
 	\end{equation*}
 	and 
 	\begin{equation*}
 	D = \gamma^{1}\otimes [H_{1},\cdot] + \gamma^{2} \otimes [H_{2},\cdot].
 	\end{equation*}
 \end{enumerate}

\begin{lemma}
	For all p+q=2 models, odd trace powers of the Dirac operator are equal to zero.
\end{lemma}
\begin{proof}
	This can be proven by showing the trace of any odd number of gamma matrices in even Clifford modules is zero. This is because in the calculation of the trace of powers of the Dirac operator, all matrix variables have a coefficient that is the trace of a product of gamma matrices
	\begin{equation*}
	\tr (\gamma^{\mu_{1}}...\gamma^{\mu_{n}}),
	\end{equation*}
	where $n\geq 3$ is odd. First suppose that all the gamma matrices are the same. Then the odd powers of skew-Hermitian matrices is itself skew-Hermitian, and therefore traceless. Now suppose that at least one of them is different then the product can be rewritten using the cyclic property of trace as 
	\begin{equation*}
	\pm\tr (\gamma^{\mu_{1}'}\gamma^{\mu_{2}'}...\gamma^{\mu_{n-2}'}\Gamma),
	\end{equation*}
	where $\Gamma$ is the chirality operator. It is a well known property that the trace of the chirality operator times an odd number of gamma matrices is zero.
\end{proof}

\begin{prop} Consider a formal Dirac ensemble of  the form 
	\begin{equation*}
		\int_{\mathcal{G}} e^{-\frac{1}{8}\tr D^{2} - \frac{t_{4}}{16}\tr D^{4} -\frac{t_{6}}{32}\tr D^{6}} dD.
	\end{equation*}
	In the large $N$ limit the underlying random matrix model is the same for all $p+q =2$ geometries.
\end{prop}
Proof of this can be seen from using the explicit formulas given in section 4 of \cite{Sanchez} and knowing that all odd trace matrix powers contribute nothing in the large $N$ limit, see Appendix A. This result may very well be true for higher powers but since no general formula for trace of powers of $D$ for this class of models is known, it is difficult to prove such a result.

 Note however, the Dirac operators here are still different which explains their distinct behaviour seen in  \cite{Barrett2016,glaser}. Furthermore, in general the relationship between a Dirac operator's spectrum and its random matrix spectrum is unclear.

\section{Bitracial Matrix Models}
In this section we analyze the bitracial single matrix models whose form originates from the (1,0) and (0,1) geometries in \cite{Barrett2016}. Examples of these model have been analyzed to some extent in both formal \cite{AK} and convergent cases \cite{First paper}. Consider  the following formal matrix integral over the space of Hermitian matrices.

\begin{equation}\label{multi trace ensemble}
Z =\int_{\mathcal{H}_{N}} e^{-V(H)} dH, 
\end{equation}
where the potential can be written as a bitracial polynomial 
\begin{equation}\label{potential}
V(H) = \frac{N}{t}\tr H^{2} - t_{1,1}\tr H\tr H - \sum_{j=3}^{d}\left(\frac{N}{jt}t_{j}\tr H^{j}+ \sum_{k=1}^{j } \frac{t_{j-k,k}}{(j-k)k}\tr H^{j-k} \tr H^{k}\right),
\end{equation}
where $t_{j}$ and $t_{\ell-k,k}$ are coupling constants and $t>0$.

We define the moments  and cumulants of the random Hermitian matrix ensemble  as 
\begin{equation*}
\mathcal{T}_{\ell} := \langle \tr H^{\ell} \rangle ,
\end{equation*} 
\begin{equation*}
\mathcal{T}_{\ell_{1},...,\ell_{k}} := \langle \tr, H^{\ell_{1}}...\tr H^{\ell_{k}} \rangle_{c}
\end{equation*}
and the connected k-point correlators 
\begin{equation*}
W_{k}(x_{1},...,x_{k}):= \sum_{\ell_{1},...,\ell_{k}=0}^{\infty} \frac{\mathcal{T}_{\ell_{1},...,\ell_{k}}}{x_{1}^{\ell_{1}+1}...x_{k}^{\ell_{k}+1}}.
\end{equation*}

We will give a brief summary in this section as to how these integrals can be used to count the number of ways to construct surfaces called stuffed maps. For more about stuffed maps see \cite{blobbed1,blobbed,blobbed2}.

\subsection{Formal bitracial matrix models and stuffed maps}  

Formal matrix models have an interpretation as being sums over various types of maps \cite{Eynard2018}. More specifically a multitrace matrix model has a graphical interpretation as a formal sum in terms of stuffed maps \cite{blobbed}.  An orientable surface of genus $g$ with $k$ boundaries of fixed lengths is called a 2-cell of topology $(k,h)$. 
\begin{figure}[H]
	\includegraphics[width=8cm]{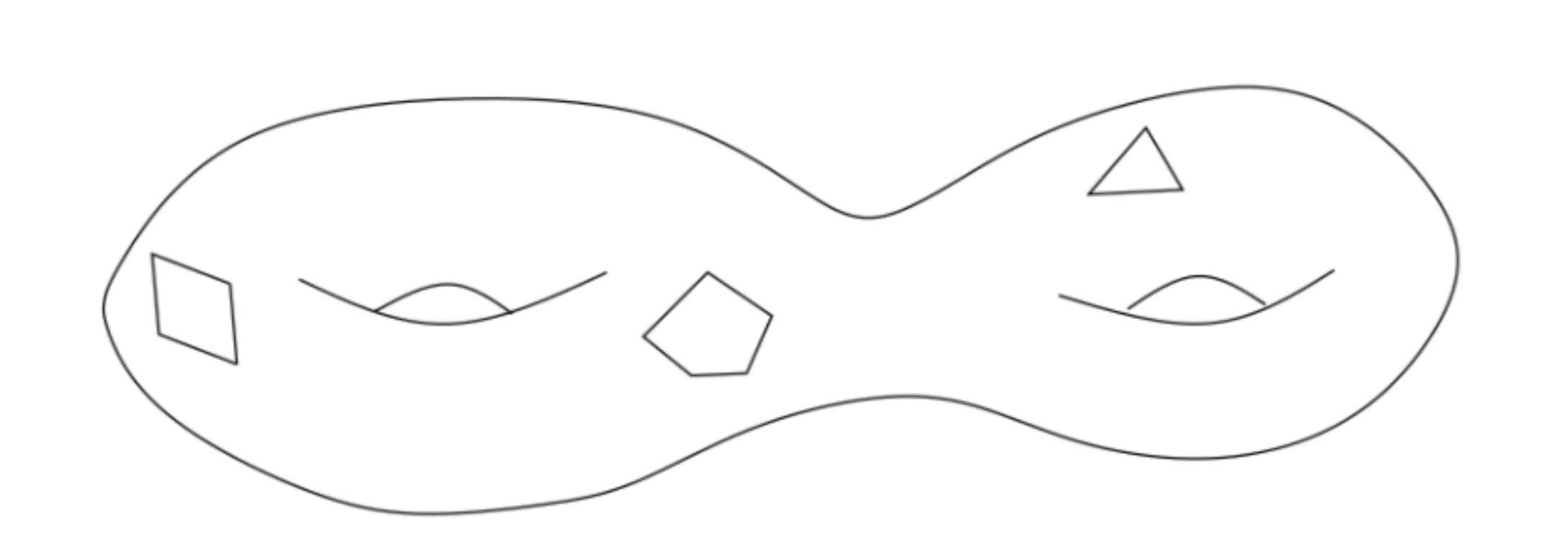}
	\centering
	\caption{An example of a genus two surface resulting from gluing various 2-cells.}
\end{figure}
When 2-cells are glued together along their edges, in an orientation preserving manner, the resulting surface is called a \text{stuffed map}. 
\begin{definition}
	A stuffed map of topology $(n,g)$ with perimeters $(\ell_{1},...,\ell_{k})$ is a genus $g$ orientable surface with $n$ marked 2-cells with the topology of  discs of lengths $(\ell_{1},...,\ell_{k})$ \cite{blobbed}.
\end{definition}

For a basic example consider a simple 2-cell with the topology of a disc and with four edges i.e. a quadrangle. It can be glued into either a map with the topology of a disc or torus.
\begin{figure}[H]
	\includegraphics[width=10cm]{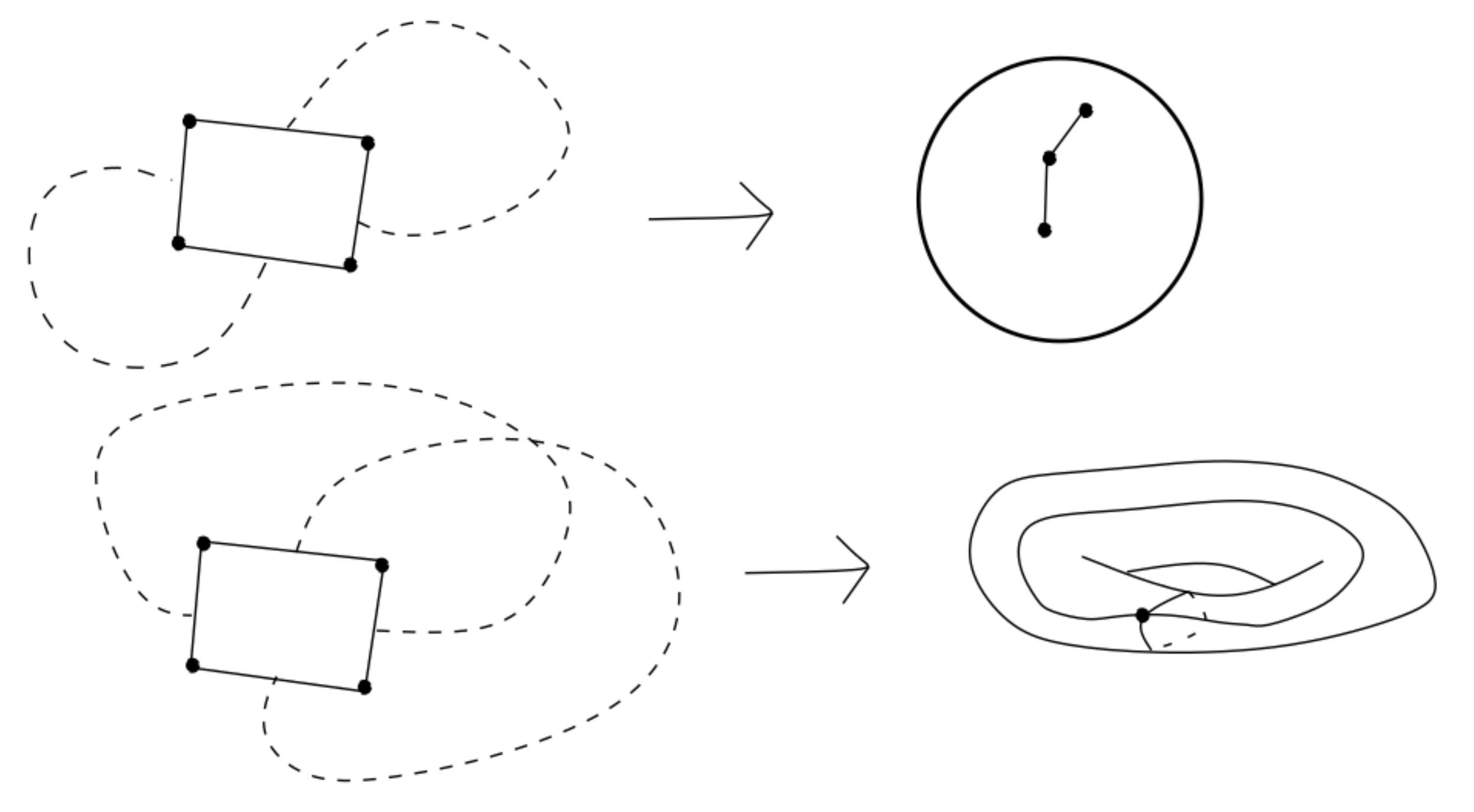}
	\centering
	\caption{An example of two possible surfaces resulting from gluing a quadrangle.}
\end{figure}

As another example, consider one possible gluing of a 2-cell with two boundaries each of length 2 with a quadrangle.
\begin{figure}[H]
	\includegraphics[scale=0.6]{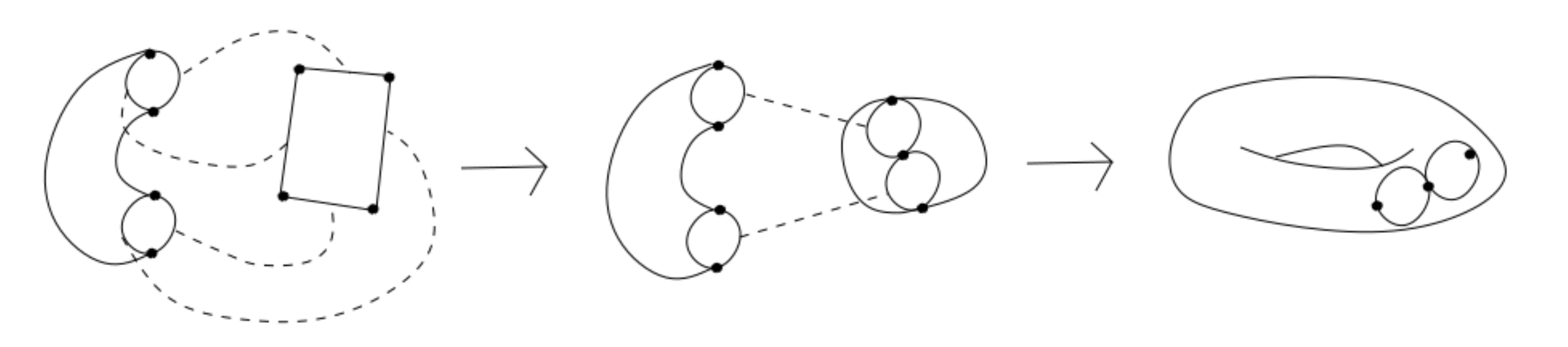}
	\centering
	\caption{An example of a possible surface resulting from gluing various 2-cells.}
\end{figure}

As one can see from the second example, counting the number of gluings by hand quickly becomes very difficult. Blobbed recursion is a beautiful tool for this task, but is not the focus of this paper. For more information about blobbed topological recursion see \cite{blobbed1, blobbed,blobbed2}.

To see how this graphical interpretation arises from matrix models consider terms of the form
\begin{equation*}
\left(\frac{N}{t}\right)^{2-2h-k}\tr H^{\ell_{1}}...\tr H^{\ell_{k}}
\end{equation*}
in the potential of a matrix model. For a fixed $h$ and $k$ there is a unique corresponding 2-cell of topology $(h,k)$. Applying Wick's theorem to compute the Gaussian expectation is graphically represented by gluing the edges of the boundaries of this 2-cell together in all possible ways \cite{Brezin1978}. Once each boundary has all its edges glued, in some orientation preserving way, we are left with a stuffed map. Wick's theorem thus tells us that we are summing over such pairings i.e stuffed maps: 
\begin{equation*}
\left(\frac{N}{t}\right)^{2-2h-k} \frac{1}{k!\ell_{1} ...\ell_{k}}\langle \tr H^{\ell_{1}}...\tr H^{\ell_{k}}\rangle_{0} = \sum_{\Sigma}\frac{t^{v(\Sigma)}}{|\text{Aut}(\Sigma)|}\left(\frac{N}{t}\right)^{\chi(\Sigma)}
\end{equation*}
where a weight $t$ is assigned to each vertex and $\chi(\Sigma) = 2-2g -k$ is the Euler characteristic of each resulting stuffed map by 't Hooft's classical argument \cite{Brezin1978,Eynard2018}. More generally we may write the expectation values of the model as
\begin{equation*}
\left\langle \prod_{i=1}^{m}\frac{1}{n_{i}!}\left(	\left(\frac{N}{t}\right)^{2-2h_{i}-k_{i}} \frac{1}{k_{i}!\ell_{1_{i}} ...\ell_{k_{i}}} \tr H^{\ell_{1_{i}}}...\tr H^{\ell_{k_{i}}}\right)^{n_{i}} \right\rangle = \sum_{\text{Stuffed Maps}\,\, \Sigma}\frac{t^{V(\Sigma)}}{|\text{Aut}(\Sigma)|}\left(\frac{N}{t}\right)^{\chi(\Sigma)},
\end{equation*} 
where the sum is over all stuffed maps (not necessarily connected) glued from $n_{i}$ 2-cells of topology $(h_{i},k_{i})$ with boundaries of lengths $\ell_{1_{i}},...,\ell_{k_{i}}$ for $1\leq i \leq m$.

\begin{definition}
	Let $\mathbb{S}\mathbb{M}^{g}_{k}(v)$ be the set of connected stuffed maps of genus $g$  and $v$ vertices glued from 
	\begin{itemize}
		\item $k$ boundaries with the topology of the disc,
		\item $n_{3}$ triangles, $n_{4}$ quandrangles, ... $n_{d}$ $d$-gons, 
		\item $m_{i,j}$ cylinders of a $j$-gon and $k$-gon such that $i+j = q$ for $ 2\leq q \leq d$, and $i \not =0$, $j\not =0$,  
		\item $\mathbb{S}\mathbb{M}^{0}_{2}(1)=\{.\}$.
	\end{itemize}

\end{definition}

In multimatrix models, different colours correspond to different matrix variables. For example 
\begin{equation*}
	Z = \int_{\mathcal{H}_{N}^{2}} e^{-
	\frac{N}{t}\left(\tr A^{4} + \tr B^{4}+\tr A B\right)} dAdB.
\end{equation*}
When this model is treated as a formal matrix model, it is a sum over the gluings of quadrangles of two possible colours and a genus zero 2-cell with two perimeters each of different colour.

For a fixed genus with a given  number of boundaries and vertices, and a given topologies of 2-cells, we wish to show that the set of all possible stuffed maps is finite. This would allow us to reorganize formal multitrace multimatrix  integrals and prove what is known as a genus expansion. With this in mind we define the following the type of map.

\begin{definition}
	An $M$-coloured stuffed map of genus $g$ with $k$ boundaries is a genus $g$ map glued from 2-cells of any topology whose boundaries' edges can be any of $M$ different colours. 
\end{definition}

We are in particular are interested in when the coloured stuffed maps are glued strictly from 2-cells with the topologies of the disc and the cylinder.

\begin{definition}
	Let $\mathbb{S}_{M}\mathbb{M}^{g}_{k}(v)$ be the set of connected stuffed maps of genus $g$  and $v$ vertices glued from 
	\begin{itemize}
		\item $k$ boundaries with the topology of the disc,
	\item $n_{3}^{r}$ triangles, $n_{4}^{r}$ quadrangles, ... $n_{d}^{r}$ $d$-gons of any of $M$ colours indexed by $r$,
	\item $m_{i,j}^{r}$ cylinders of a $j$-gon and $k$-gon of any of $M$ colours indexed by $r$, for $ 2\leq i,j \leq d$, such that $i \not =0$ and $j\not =0$,
	\item $\mathbb{S}_{M}\mathbb{M}^{0}_{2}(1)=\{.\}$.
	\end{itemize}
\end{definition}

\begin{theorem}\label{finite set}
	The set $\mathbb{S}_{M}\mathbb{M}^{g}_{k}(v)$ of all maps described above is finite.
\end{theorem}

Consider an elementary  2-cell $C_{g,k}$ with genus $g$ and Euler characteristic strictly less than one with $k$ boundaries that are not connected by edges. When it is glued as part of a stuffed map it acts as a bridge between at most $k$ connected graphs embedded into a surface. Call these graphs the graph components of the stuffed map. Note that a 'usual map' (i.e only glued from 2-cells with the topology of a disc) has only one graph component. This concept is key to the proof below.
\begin{figure}[H]
	\includegraphics[width=8cm]{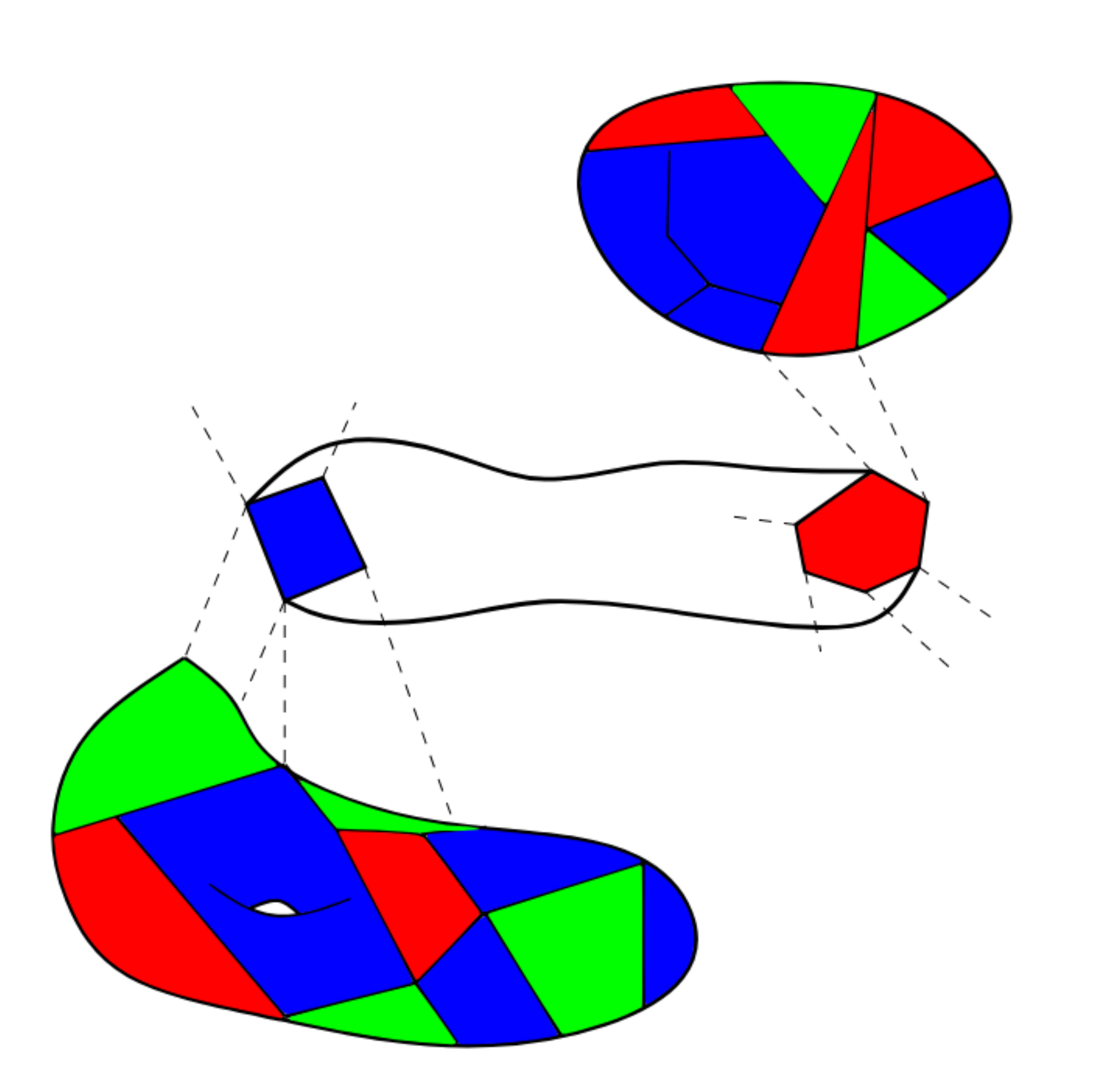}
	\centering
	\caption{The center 2-cell acts as a bridge between  the other two graphs created by the edges of the various coloured polygons.}
\end{figure}

\begin{proof}

	Consider an $M$-coloured stuffed map of genus $g$ with $k$ boundaries and $v$ vertices. The number of graph components $q$ is at most $v$ since each component, call them $C_{i}$, must have at least one vertex. The total genus of the stuffed map $g$ is  the sum of the genus $g_{i}$ of the $i$th component, the genus of two cells and the amount of handles created by the bridges connecting all components, $g_{B}$. The number of boundaries (i.e marked polygons) $k_{i}$ on each component must total $k$.
	If we remove the bridges and only keep the boundaries of all the 2-cells glued, the result are $q$ connected 'usual' maps $C_{i}$ glued from polygons, with $k_{i}$ boundaries and $g_{i}$ handles.

	 For each map we know by Euler's theorem 
	\begin{equation*}
	2-2g_{i} = k_{i}+f_{i}-e_{i} +v_{i}.
	\end{equation*}
	The number of faces is given by
	\begin{equation*}
		f_{i} = \sum_{j=1}^{d} n_{j}^{i}
	\end{equation*}
	where $n_{j}^{i}$ denotes the number of unmarked faces of the graph $C_{i}$ with length $j$. The number of edges is equal to
	\begin{equation*}
	e_{i} = \frac{1}{2}\left(\sum_{j=1}^{k_{i}}\ell_{j}(C_{i})+ \sum_{j=1}^{d}jn_{j}^{i} \right).
	\end{equation*} 
	This allows us to write
	
	\begin{equation*}
	v_{i}-2+2g_{i}+k_{i} = \frac{1}{2}\left(\sum_{j=1}^{k_{i}}\ell_{j}(C_{i})+ \sum_{j=3}^{d}(j-2)n_{j}^{i} + \sum_{j=1}^{2}(j-2)n_{j}^{i}\right).
	\end{equation*}
	Since when $j \geq 3$, we have $i-2 \geq 1$, we write
		\begin{equation*}
	v_{i}-2+2g_{i}+k_{i} + \frac{1}{2}n^{i}_{1} \geq \frac{1}{2}\left(\sum_{j=1}^{k_{i}}\ell_{j}(C_{i})+ \sum_{j=3}^{d}n_{j}^{i} \right).
	\end{equation*}
	
	Each variable for a fixed $C_{i}$ is fixed in left hand side except maybe $n_{1}^{i}$, hence on each component the number of faces with degree greater than or equal to three is finite if $n_{1}^{i}$ is finite. 
	
	Let any 2-cell with a degree one or two on a component be called a strip. Hence, for each $C_{i}$ polygons that are of degree one or two can only belong to a strip since all 2-cells with one boundary have a minimum length of three. We claim for each $C_{i}$ there are only finitely many strips attached and hence finitely many strips in total . 
	
	A bridge can do two possible things, it either connects a graph component to a new graph component or a graph component to itself (see picture). The number of bridges that connect to different graph components must be finite since the number of graph components is finite. This is because the number of vertices is fixed and each new graph component has at least one vertex. Each bridge of the latter type increases the genus of the stuffed map which is bounded by $g$. Thus the number of ways to glue bridges is finite.

	Hence, the number of strips and therefore $n^{i}_{1}$ and $n^{i}_{2}$ on each component is finite and since there are finitely many components, this completes the proof. We also know that $n^{i}_{1} \leq 2q$ for all $C_{i}$. Thus we find a useful inequality by summing the above equalities for all graph components:
	
	\begin{equation*}\label{Euler equality}
	v + 2g + k +
	q\geq   \frac{1}{2}\left(\sum_{j=1}^{k}\ell_{j}+ \sum_{i=1}^{q}\sum_{j=3}^{d}(j-2)n_{i,j}\right),
	\end{equation*}
	and since $q \leq v $ we have 
	\begin{equation*}\label{Euler equality2}
	2v + 2g + k +
	\geq   \frac{1}{2}\left(\sum_{j=1}^{k}\ell_{j}+ \sum_{i=1}^{q}\sum_{j=3}^{d}(j-2)n_{i,j}\right).
	\end{equation*}
	
\end{proof}

\begin{corollary}
	$\mathbb{S}\mathbb{M}^{g}_{k}(v)$ is a finite set.
\end{corollary}

This result will be later used to prove that a wide class of multitrace random matrix models 
satisfy Brown's lemma, validating the assumptions made in \cite{AK}. Furthermore, in a similar manner as in \cite{Eynard2018}, the last inequality tells us that $v+2g+k\geq 0$ and since the number of maps for a fixed v, g, and k is finite we are able to define the formal power series that appears in the following theorem.
 
\begin{theorem}\label{genus expansio }
	Let $S$ be a real monic polynomial in $m$-variables with powers of $N$ in the coefficients, that is symmetric in each variable. Consider a formal matrix  integral of the form
	\begin{equation*}
		\int_{\mathcal{H}_{N}^{m}} e^{-V( H_{1},  H_{2},..., H_{m})} dH_{1}...dH_{m},
	\end{equation*}
	where the potential is a multitrace polynomial.
	
	Define the $k$-resolvent function to be 
	\begin{equation*}
		W_{k}(x_{1},...,x_{k}) = \sum_{\ell_{1},...,\ell_{k}}^{\infty}\frac{\langle \tr H_{1}^{\ell_{1}}...\tr H_{1}^{\ell_{k}}\rangle_{c}}{x_{1}^{\ell_{1}+1} ...x_{k}^{\ell_{k}+1}}.
	\end{equation*}
	Then the $k$-resolvent has a genus expansion
	\begin{equation*}
	W_{k}(x_{1},...,x_{k}) = \sum_{g=0}^{\infty}\left(\frac{N}{t}\right)^{2-2g-k} W^{g}_{k}(x_{1},...,x_{k}).
	\end{equation*}   
\end{theorem}

This is because for the following formal summations:

\begin{equation*}
\mathcal{T}_{\ell_{1},...,\ell_{k}}^{g}:= \sum_{v=1}^{\infty}t^{v} \sum_{\Sigma \in \mathbb{S}_{M}\mathbb{M}^{g}_{k}(v)}   t_{3}^{n_{3}(\Sigma)}... t_{d}^{n_{d}(\Sigma)} \frac{t^{V(\Sigma)}}{|\text{Aut}(\Sigma)|}\prod_{i=1}^{k}\delta_{\ell_{i},\ell_{i}(m)},
\end{equation*} 

\begin{equation*}
W_{k}^{g}:= \sum_{v=1}^{\infty}t^{v} \sum_{\Sigma \in \mathbb{S}_{M}\mathbb{M}^{g}_{k}(v)} \frac{ t_{3}^{n_{3}(\Sigma)}... t_{d}^{n_{d}(\Sigma)}}{x_{1}^{\ell_{1}+1}...x_{d}^{\ell_{d}+1}} \frac{t^{V(\Sigma)}}{|\text{Aut}(\Sigma)|}
\end{equation*} 
we have that to any order in $t$, the above inequality implies that the sum over $g$ is finite. This allows us to define generating functions that disregard the genus, i.e.
\begin{equation*}
\mathcal{T}_{k} = \sum_{g=0}^{\infty}\left(\frac{N}{t}\right)^{2-2g-k}\mathcal{T}_{k}^{g},
\end{equation*}
and
\begin{equation*}
W_{k}(x_{1},...,x_{k})  = \sum_{g=0}^{\infty}\left(\frac{N}{t}\right)^{2-2g-k}W_{k}^{g}(x_{1},...,x_{k}).
\end{equation*}

\subsection{Loop Equations}
All matrix models satisfy a set of equations that relate their moments and cumulants. These equations were derived in \cite{blobbed} and \cite{AK} for formal multitrace models. First let us rewrite the potential from equation (\ref{potential}) as 

\begin{equation*}
	V(H) =  \frac{N}{t}\tr H^{2}  - \sum_{j=2}^{d}\left(\frac{N}{jt}\tilde{t}_{j}\tr H^{j}\right),
\end{equation*}
where $\tilde{t}_{j}$'s include the appropriate  $t_{j}$'s and tracial moments. Then the first loop equation becomes 
\begin{equation*}
	\left(W_{1}^{0}(x)\right)^{2} = V'(x)W_{1}^{0}(x) - P_{1}^{0}(x),
\end{equation*}  
where 
\begin{equation*}
	P_{1}^{0}(x) = t - \sum_{j=2}^{d}\sum_{t = 0}^{j-2}\tilde{t}_{j} \mathcal{T}^{0}_{j-\ell-2}x^{\ell}.
\end{equation*}
See \cite{Eynard2018} for details. It is clear that we may write 
\begin{equation*}
	W_{1}^{0}(x) = \frac{1}{2}\left(S'(x) - \sqrt{V'(x)^{2} - 4 P_{1}^{0}(x)}\right).
\end{equation*}

We now generalize a famous lemma from \cite{Eynard2018} that simplifies this expression.

\begin{lemma}[1-Cut Brown's Lemma]
	There exists formal powers series $\alpha$, $\gamma^{2}$, and a polynomial $M(x)$ such that 
	\begin{equation*}
		\alpha = O(t), \qquad \gamma^{2} = t + O(t^{2}), \qquad M(x)= \frac{V'(x)}{x} + O(t),
	\end{equation*}
	and 
	\begin{equation*}
		V'(x)^{2} - 4 P_{1}^{0}(x) = (M(x))^{2}(x-a)(x-b)
	\end{equation*}
	with $a = \alpha +2\gamma$ and  $b = \alpha - 2\gamma$.
\end{lemma}
\begin{proof}
	This proof is identical to the proof of lemma 3.1.1 in \cite{Eynard2018}, except for replacing the sum over maps with stuffed maps and inequality 3.1.3 with  an analogous one from the proof of theorem (\ref{finite set}).
\end{proof}

This is a rather technical lemma with many auxiliary formal series, but the take away is the factorization of $S'(x)^{2} - 4 P_{1}^{0}(x)$.

\subsection{Convergent bitracial matrix models} Consider  the following convergent matrix integral over the space of Hermitian matrices

\begin{equation*}
Z =\int_{\mathcal{H}_{N}} e^{-V(H)} dH,
\end{equation*}
where the potential can be written as a multitrace polynomial
\begin{equation*}
V(H) = \frac{N}{2t}\tr H^{2} - t_{1,1}\tr H\tr H - \sum_{j=3}^{d}\left(\frac{N}{jt}t_{j}\tr H^{j}+\frac{1}{2} \sum_{k=1}^{j } \frac{t_{j-k,k}}{(j-k)k}\tr H^{j-k} \tr H^{k}\right),
\end{equation*}
where $t_{\ell-k,k}$ are coupling constants in ranges where this model is convergent. This method was used \cite{First paper} and generalize to higher order models. We will summarize this here.
This model is invariant under the action of the unitary group on the Hermitian matrix $H$, allowing us to apply Weyl's integration formula  to write

\begin{equation*} \label{jpd}
Z = C_{N} \int_{\mathbb{R}^{N}} e^{-N\sum_{i=1}^{N}Q(\lambda_{i}) - \sum_{i,j=1}^{N}U(\lambda_{i}, \lambda_{j})} \prod_{1\leq i  < j \leq N}(\lambda_{i}-\lambda_{j})^{2}d\lambda_{1}...\lambda_{N},
\end{equation*}
where
\begin{equation*}
Q(x) = \frac{1}{2t} x^{2} - \sum_{j=3}^{d} \frac{t_{j}}{jt}x^{j},
\end{equation*}
and
\begin{equation*}
U(x,y) = -t_{1,1}xy -\frac{1}{2}\sum_{j=3}^{d}\sum_{k=1}^{j}\frac{t_{j-k,k}}{(j-k)k} x^{j-k}  y^{k},
\end{equation*}
and $C_{N}$ is a constant. From here it is explained in \cite{First paper} how to compute the limiting spectral distribution of eigenvalues using the Euler-Lagrange equations. 

 For convergent matrix models a factorization of the form seen above can often be found. However, we are not aware of a proof of its existence in general. Thus the results of this paper can only be applied to convergent models on a case by case basis.

\section{Moment Generating Functions of Gaussian Dirac Ensembles} \label{proof}
	
	Consider a Dirac ensemble of type (p,q) geometry, where the gamma matrices act on $\mathbb{C}^{k}$. Let $D$ be the Dirac operator on that space with a Gaussian potential, i.e.
\begin{equation*}
Z = \int_{\mathcal{G}}e^{-\frac{1}{2k}\tr D^{2}}dD.
\end{equation*}

While the analytic study of general Dirac ensembles is a  very difficult task, we can say a fair bit about the Gaussian case that is nontrivial and universal. From \cite{Barrett2015} we know such a Dirac operator is of the form 
\begin{equation*}
D = \sum_{j} \alpha_{i}\otimes [L_{j},\cdot]+ \sum_{k} \beta_{k}\otimes\{H_{k},\cdot\}+ \sum_{\ell} \alpha_{\ell}'\otimes\{L_{\ell},\cdot\}+ \sum_{r} \beta_{\ell}'\otimes[H_{r},\cdot],
\end{equation*}
where the products of gamma matrices all belong to a linearly independent set of matrices. Now consider $D^{2}$. Each term of $D^2$ consists of  two linearly independent matrices tensored with some commutator or anticommutator of some skew-Hermitian or Hermitian random matrix. Using Proposition 3.5 of \cite{Sanchez} we know
\begin{align*}
\tr D^{2} &= \sum_{j} \tr \alpha_{i}^{2}\tr  [L_{j},\cdot]^{2}+ \sum_{q} \tr \beta_{q}^{2}\tr\{H_{q},\cdot\}^{2}+ \sum_{\ell} \tr\alpha_{\ell}'^{2}\tr\{L_{\ell},\cdot\}^{2}+ \sum_{r} \tr \beta_{\ell}'^{2}\tr[H_{r},\cdot]^{2}\\
&= 2\sum_{j} \tr \alpha_{i}^{2}  \left(-N\tr L_{j}^{2}+(\tr L_{j})^{2}\right)+ 2\sum_{q} \tr \beta_{q}^{2}\left(N\tr H_{q}^{2}+(\tr H_{q})^{2}\right)\\
&- 2\sum_{\ell} \tr\alpha_{\ell}'^{2}\left(N\tr L_{\ell}^{2} + (\tr L_{\ell})^{2}\right)+ 2\sum_{r} \tr \beta_{\ell}'^{2}(-N\tr H_{r}^{2} + (\tr H_{r})^{2}).
\end{align*}


and 
\begin{align*}
&= 2\sum_{j} k \left(-N\tr L_{j}^{2}+(\tr L_{j})^{2}\right)+ 2\sum_{q} k\left(N\tr H_{q}^{2}+(\tr H_{q})^{2}\right)\\
&- 2\sum_{\ell} k\left(N\tr L_{\ell}^{2} + (\tr L_{\ell})^{2}\right)+ 2\sum_{r} k(-N\tr H_{r}^{2} + (\tr H_{r})^{2}).
\end{align*}
Skew-Hermitian matrices are traceless so the above sum is equal to 
\begin{align*}
& =2\sum_{j} k \left(-N\tr L_{j}^{2}\right)+ 2\sum_{q} k\left(N\tr H_{q}^{2}+(\tr H_{q})^{2}\right)\\
&- 2\sum_{\ell} k\left(N\tr L_{\ell}^{2} \right)+ 2\sum_{r} k(-N\tr H_{r}^{2} + (\tr H_{r})^{2}).
\end{align*}
Any skew-Hermitian matrix can be written as $i$ times a Hermitian matrix. Making this substitution  gives us
\begin{align*}
\frac{1}{2k} \tr D^{2} &= \sum_{j}  \left(N\tr H_{j}^{2}\right)+ \sum_{q}\left(N\tr H_{q}^{2}+(\tr H_{q})^{2}\right)\\
& + \sum_{\ell} \left(N\tr H_{\ell}^{2} \right)+ \sum_{r} (-N\tr H_{r}^{2} + (\tr H_{r})^{2}).
\end{align*}

Next apply the transformation $H_{r}\rightarrow \sqrt{-1} H_{r}$ for each $r$ to get
\begin{align*}
\frac{1}{2k} \tr D^{2} &= \sum_{j}  \left(N\tr H_{j}^{2}\right)+ \sum_{q}\left(N\tr H_{q}^{2}+(\tr H_{q})^{2}\right)\\
& +\sum_{\ell} \left(N\tr H_{\ell}^{2} \right)+ \sum_{r} (N\tr H_{r}^{2} - (\tr H_{r})^{2}).
\end{align*}

Hence, we have realized the partition function as a finite product of matrix integrals:

\begin{equation*}
Z = \int_{\mathcal{G}}e^{-\frac{1}{2k}\tr D^{2}}dD = c \prod_{\mu}\left( \int_{\mathcal{H}_{N}}e^{-N\tr H_{\mu}^{2} }dH_{\mu}\right)\prod_{\nu}\left( \int_{\mathcal{H}_{N}}e^{-N\tr H_{\nu} \pm {(\tr H_{\nu}^{2})^{2}}}dH_{\nu}\right),
\end{equation*}
where the constant $c$ is some power of $i$ determined by the number of transformations used above. Since the above integral is separable in terms of its matrix variables, the covariance matrix of this model is a block diagonal matrix and thus the correlation between two different matrix variables is zero giving us the following lemma.

\begin{lemma}\label{correlation}
	The partition function of a Gaussian Dirac ensemble  can be decomposed as a product of random Hermitian matrices such that the correlation between different matrix variables is zero, i.e.
	
	\begin{equation*}
	\langle \tr H_{\mu}^{n}H_{\nu}^{m} \rangle = 0,
	\end{equation*}
	for $\mu \not = \nu,$ and all $m, n$. 
\end{lemma}

Now refer to Gaussian multitrace example in Appendix A. It is clear that in the large $N$ limit the $(\tr H)^{2}$ contribution is zero. Hence, in the limit $Z$ becomes the product of Gaussians, all with identical spectral statistics. 
With this in mind and the above lemma we will prove the following.

\begin{prop}
	For any Gaussian Dirac ensemble 
	\begin{equation*}
	\lim_{N \rightarrow \infty}\frac{1}{2k} \langle \tr D^{m} \rangle =  \sum_{j=0}^{m}{m \choose j} \mathcal{T}^{0}_{m-j}\mathcal{T}_{j}^{0},
	\end{equation*}
	where $\mathcal{T}_{j}^{0}$ denotes the $j$ Gaussian moment in the large $N$ limit, which are well known to be the Catalan numbers.
\end{prop}
Note that this implies that all odd moments are zero since all odd Gaussian moments are zero.
\begin{proof}
	Consider $\tr D^{m}$ of any Gaussian Dirac ensemble where we consider the integral just in terms of Hermitian random matrices by  using the substitutions mentioned above. Now consider  $ \langle \tr D^{m} \rangle$. We know that by Lemma \ref{correlation}  all mixed terms in $ \langle \tr D^{m} \rangle$ will vanish. Furthermore the remaining terms are all Gaussian moments, so for $m$ odd this whole sum of Gaussian terms vanishes. Now consider the case when $m$ is even. Once again mixed terms vanish in the limit  by Lemma \ref{correlation}, leaving only powers of anticommutators and commutators:

	\begin{align*}
	\lim_{N \rightarrow \infty}\frac{1}{2k} \langle \tr D^{m} \rangle &= \lim_{N \rightarrow \infty}\frac{1}{2k} \left( \sum_{j}  \tr \alpha_{j}^{m} \langle \tr [L_{j},\cdot]^{m}\rangle + \sum_{q} \tr \beta_{q}^{m} \langle \tr \{H_{q},\cdot\}^{m}\rangle\right.\\
	& \left.+ \sum_{\ell} \tr (\alpha_{\ell}')^{m}\langle \tr  \{L_{\ell},\cdot\}^{m}\rangle+ \sum_{r} \tr (\beta_{\ell}')^{m}\langle \tr[H_{r},\cdot]^{r}\rangle \right).
	\end{align*}
	 Recall from \cite{Barrett2015}  that these powers of products of gamma matrices, call them $w_{i}$, are either Hermitian or skew-Hermitian, depending on whether they are tensored with a Hermitian or skew-Hermitian random matrix variable. Thus after the substitution $H =iL$, each gamma matrix power has order two. Hence, trace we have

	\begin{align*}
	\tr D^{m} &=  \sum_{q} \tr \beta_{q}^{m}\tr\{H_{q},\cdot\}^{m}+ \sum_{r} \tr \beta_{\ell}'^{m}\tr[H_{r},\cdot]^{m}\\
	&=k\left(\sum_{q} \tr\{H_{q},\cdot\}^{m}+ \sum_{r}  \tr[H_{r},\cdot]^{m}\right)\\
	&= k\left(\sum_{q} \sum_{j=0}^{m}{m \choose j}H_{q}^{j}H_{q}^{m-j}+ \sum_{r} \sum_{j=0}^{m} (-1)^{j+1}{m \choose j}H_{r}^{j}H_{r}^{m-j}\right).\\
	\end{align*}
	Recall that the large $N$ limit of the expectation value of the terms in the above sum that have odd moments go to zero by the symmetry of the model, thus completing the proof.
	
\end{proof}
Define 
\begin{equation*}
 \zeta_{\ell}^{0}:=\lim_{N\rightarrow\infty}\langle \tr D^{\ell}\rangle = \sum_{k=0}^{\ell} {\ell \choose k} \mathcal{T}^{0}_{\ell-k}\mathcal{T}^{0}_{k} .
\end{equation*}
Consider the following Dirac exponential generating function (DEGF):
\begin{equation*}
\mathcal{D}(x):=	\lim_{N\rightarrow\infty}\langle \tr e^{xD}\rangle= \sum_{\ell =  0}^{\infty}\frac{\zeta^{0}_{\ell}}{\ell!}x^{\ell}.
\end{equation*}
Multiply the above equation by $x^{\ell}/\ell!$ and sum from zero to infinity and we find that the DEGF is the square of the matrix exponential generating function (MEGF)
\begin{equation*}
G(x)^{2} : = \left( \sum_{\ell=0}^{\infty}\frac{\mathcal{T}^{0}_{\ell}}{\ell!}x^{\ell}\right)^{2 } =\mathcal{D}(x).
\end{equation*}

It is well known that the exponential moment generating function of the Wigner Semicircle distribution is $I_{1}(2x)/x$,  where $I_{1}$ denotes the modified Bessel function of the first kind. Hence, this completes the proof of the main result. Thus it follows from above Proposition that the  limit 
\begin{equation*}
\sqrt{\lim_{N\rightarrow \infty} \langle \tr e^{x D}\rangle}  = \frac{I_{1}(2x)}{x},
\end{equation*}
exists.

\begin{theorem}[Wigner Convolution Law] \label{Convolution law}
	For any Gaussian Dirac ensemble, the limiting spectral density function of the transformed Dirac operators is 
	\begin{equation*}
	\rho_{D}(x) = \int_{\mathbb{R}}\rho_{W}(x-t)\rho_{W}(x)dt,
	\end{equation*}
	where
	\begin{equation*}
	\rho_{W}(x) = \frac{1}{2\pi}\sqrt{4-x^{2}}_{[-2,2]}, 
	\end{equation*}
	is the Wigner Semicircle Distribution.
\end{theorem}
\begin{proof}
	In probability theory it is well known that the moment generating function of a random variable is the two-sided Laplace transform of the probability density function. In our case from the above Proposition we see that the limiting Dirac eigenvalue distribution's moment generating function is the square of the GUE's generating function. Thus, by the convolution property of the Laplace transform, we deduce the above result. 
\end{proof}

\section{Moments and generating functions of  one matrix Dirac ensembles}
In this Section we generalize Theorem \ref{Convolution law} to non-Gaussian Dirac ensembles for $p+q = 1$.

We know from Theorem \ref{genus expansio } that the  random matrix moments of a formal multitrace model have a genus expansion. Thus the same may be said about Dirac ensemble moments:

\begin{equation*}
\langle \tr D^{\ell} \rangle= \sum_{g=0}^{\infty}N^{2-2g}\zeta_{\ell}^{g}  = \sum_{g=0}^{\infty}N^{2-2g}\sum_{j=0}^{\ell}{\ell \choose j}\sum_{h=0}^{g}\left(\mathcal{T}^{g-h}_{j}\mathcal{T}^{h}_{\ell-j}+ \mathcal{T}^{g-1}_{j,\ell-j}\right).
\end{equation*}
We know that based on the genus expansion of moments, the correlation between mixed moments vanishes in the planar expansion. We wish to compute the DEGF. Suppose one can compute $W_{1}^{0}(x)$ i.e. the resolvent matrix moment generating function from previous sections. Then we can find the matrix moment ordinary generating function
\begin{equation*}
O(x):=\frac{1}{x}W_{1}^{0}(\frac{1}{x}) = \frac{1}{x}\sum_{\ell=0}^{\infty}\mathcal{T}^{0}_{\ell}x^{\ell+1} ,
\end{equation*}
which we will eventually convert into the matrix moment exponential generating function. The moments can be extracted using Cauchy's integral formula 
\begin{equation*}
\mathcal{T}^{0}_{\ell} = \frac{1}{\ell!}\frac{\partial}{\partial x^{\ell}} O(x)|_{x=0} = \frac{1}{2\pi i}\int_{|z| = R} \frac{O(w)}{w^{\ell+1}}dw.
\end{equation*}
Thus the relation between the moment generating function (and hence the resolvent) and exponential moment generating function can be expressed as follows:
\begin{equation*}
G(x) = \frac{1}{2\pi i}\int_{|z| = R} \frac{O(w)}{w}e^{x/w}dw = \frac{1}{2\pi i}\int_{|z| = R} \frac{W^{0}_{1}(1/w)}{w^{2}}e^{x/w}dw.
\end{equation*}
The one cut lemma allows us to write the above as
\begin{align*}
&=\frac{1}{4\pi i}\int_{|z| = R} \frac{1}{w^{2}}\left(S'(1/w)-\frac{M(1/w)}{w}\sqrt{(1-aw)(1-bw)}\right)e^{x/w}dw\\
&=\frac{1}{2}\text{Res}[\frac{e^{x/w}}{w^{2}}S'(1/w),0] - \frac{1}{2}\text{Res}[\frac{e^{x/w}}{w^{3}}M(1/w)\sqrt{(1-aw)(1-bw)},0]\\
&=\frac{1}{2}\text{Res}[-\frac{1}{w^{3}}M(1/w)\sqrt{(1-aw)(1-bw)},0].
\end{align*} 
Suppose now that $\alpha = 0$, i.e. $a=-b$. This for example will always happen when the models potential is even, see 3.1.4 of \cite{Eynard2018}. Then the Laurent expansion looks like
\begin{equation*}
-\frac{M(1/w)}{w^{3}}\left(\sum_{k=0}^{\infty}{1/2 \choose k } (-a^{2}w^{2})^{k}\right)\left(\sum_{q=0}^{\infty}\frac{1}{q!}\left(\frac{x}{w}\right)^{q}\right)
\end{equation*}
\begin{equation*}
= -M(1/w)\sum_{k,q=0}^{\infty}{1/2 \choose k}\frac{(-a^{2})^{k}}{q!} x^{q}w^{2k-q-3}.
\end{equation*}

For a given model we know that $M(x)$ is a degree $d-2$  polynomial where $d$ is the degree of the potential \cite{Eynard2018}. Let
\begin{equation*}
M(1/w) := \sum_{p=0}^{d-2}q_{p}w^{-p}.
\end{equation*}
Thus we wish to compute
the residue of 
\begin{equation*}
 -\sum_{k,q=0}^{\infty}{1/2 \choose k}\frac{(-a^{2})^{k}}{q!} x^{q}w^{2k-q-3-p},
\end{equation*}
at zero for various $p$. Setting $2k-q-3 -p =-1 $, we obtain
\begin{equation}\label{MGF}
 -\sum_{k=1+p/2}^{\infty}{1/2 \choose k}\frac{(-a^{2})^{k}}{(2k-2-p)!}x^{2k-2-p} = \sum_{k=1+p/2}^{\infty}g^{p}_{k}(-a^{2})^{k}x^{2k-2-p}=\frac{1}{x^{2+p}} \sum_{k=1+p/2}^{\infty}g^{p}_{k}(a\sqrt{-1}x)^{2k}.
\end{equation}
Re-indexing equation (\ref{MGF}) we find it is equal to
\begin{equation*}
 \frac{1}{x^{p+2}}\sum_{k=p/2}^{\infty}g_{k+1}^{p}(a\sqrt{-1}x)^{2k},
\end{equation*}
where
\begin{equation*}
g_{k}^{p} = -{1/2 \choose k}\frac{1}{(2k-2-p)!} = {2k \choose k}\frac{(-1)^{k+1}}{2^{2k}(2k-1)(2k-2-p)!} = \frac{(-1)^{k+1}}{k!k! 2^{2k}}\frac{(2k)(2k-2)!}{(2k-2-p)!}
\end{equation*}
\begin{equation*}
=\frac{(-1)^{k+1}}{k!(k-1)! 2^{2k-1}}\frac{((2k-2)!}{(2k-2-p)!} = \frac{(-1)^{k+1}}{k!(k-1)! 2^{2k-1}} (2k-2)(2k-3)...(2k-p+1).
\end{equation*}
Hence,
\begin{equation*}
g_{k+1}^{p} = \frac{(-1)^{k}}{k!(k+1)! 2^{2k+1}} (2k)(2k-1)...(2k-p+3),
\end{equation*}
and 
\begin{equation*}
G(x) = \sum_{p=0}^{d-2}q_{p}\frac{1}{x^{2+p}}\sum_{k=p/2}^{\infty}g_{k+1}^{p}(a\sqrt{-1}x)^{2k+2} =\sum_{p=0}^{d-2}q_{p}\frac{a\sqrt{-1}}{x^{p+1}}\sum_{k=p/2}^{\infty}g_{k+1}^{p}(a\sqrt{-1}x)^{2k+1}.
\end{equation*}
When $p=0$ they are the $k$-th coefficients of the series expansion of the Bessel function of the first kind $J_{1}(x)$. When $p$ is larger than one, the moment generating function can still be expressed in terms of Bessel functions. Consider
\begin{equation*}
x^{p+1}\frac{d^{p}}{dx^{p}}\left(\frac{1}{x}J_{1}(x)\right) = x^{p+1}\sum_{k=p}^{\infty}\frac{(-1)^{k}}{k!(k+1)!}\frac{x^{2k-p}}{2^{2k+1}} (2k)(2k-1)...(2k-p+3)
\end{equation*}
\begin{equation*}
= \sum_{k=p}^{\infty} g_{k+1}^{0} x^{2k+1}.
\end{equation*}
The left hand side can be further simplified using the following well-known Bessel function identities:
\begin{enumerate}
	\item $J_{p}(x) = (-1)^{p}J_{-p}(x)$,
	\item $\frac{1}{x^{p}}\frac{d^{p}}{dx^{p}}\left(x^{\alpha}J_{\alpha}(x)\right) = x^{\alpha -p}J_{\alpha-p}$ 
\end{enumerate}  
for all integers $p$ and $\alpha$. 
Then we may write  
\begin{align*}
x^{p+1}\frac{d^{p}}{dx^{p}}\left(\frac{1}{x}J_{1}(x)\right)&=-x^{2p+1}\frac{1}{x^{p}}\frac{d^{p}}{dx^{p}}\left(\frac{1}{x}J_{-1}(x)\right)\\
 &= -x^{2p+1}(x^{-1-p}J_{-1-p}(x))\\ &= -x^{p}J_{-p-1}(x)\\ &= (-1)^{p}x^{p}J_{p+1}(x).
\end{align*}
This gives us 
\begin{equation*}
\sum_{k=p}^{\infty} g_{k+1}^{p} x^{2k+1} =(-1)^{p}x^{p}J_{p+1}(x),  
\end{equation*}
so
\begin{equation*}
\sum_{k=p/2}^{\infty}g^{p}_{k+1}x^{2k+1}= (-1)^{p}x^{p}J_{p+1}(x) + \sum_{k=p/2}^{p}g^{p}_{k+1}x^{2k+1}.
\end{equation*}
Finally we may express the moment generating function as
\begin{equation*}
G(x) = \sum_{p=0}^{d-2}\frac{a\sqrt{-1}q_{k}}{x^{p+1}}\left[(-1)^{p+1}(ax)^{p}J_{p+1}(a\sqrt{-1}x) + \sum_{k=p/2}^{p}g^{p}_{k+1}(a\sqrt{-1}x)^{2k+1}\right].
\end{equation*}	

Data availability statement: data sharing is not applicable to this article as no new data were created or analyzed in this study.

\section*{Appendix A: Simplifications in the large $N$ limit}
We define the Zhukovsky Transform $x:\mathbb{C}\setminus \{0\} \rightarrow \mathbb{C}$ as 
\begin{equation*}
x(z)
= \frac{a+b}{2} + \frac{a-b}{4}\left(z + \frac{1}{z}\right) = \alpha + \gamma ( z+ \frac{1}{z}),
\end{equation*}
with an inverse
\begin{equation*}
z = \frac{1}{2\gamma}\left(x -\alpha + \sqrt{(x-\alpha)^{2}-4\gamma^{2}}\right).
\end{equation*}
It also has the following useful identity
\begin{equation*}
\sqrt{(x(z)-a)(x(z)-b)} = \frac{a-b}{4}\left(z-\frac{1}{z}\right).
\end{equation*}

Using Theorem 3.1.1 of \cite{Eynard2018}, we have the following result.
\begin{prop}
	For any formal power series $\alpha$ and $\gamma$ as mentioned in the One-cut Lemma, we have the expansions 
	\begin{equation*}
	V'(x(z)) = \sum_{k=0}^{d-1}u_{k}(z^{k}+z^{-k})
	\end{equation*}
	and 
	\begin{equation*}
	W_{1}^{0}(x(z)) = \sum_{k=0}^{d-1}u_{k}z^{-k}
	\end{equation*}
	with $u_{0} = 0$ and $u_{1} = t/\gamma$.
\end{prop} 
This shows the relationship between the resolvent and the potential function $V(x)$ in Zhukovsky coordinates. The details  are dealt with in \cite{AK,Eynard2018}.

Now consider the integral
\begin{equation*}
Z = \int_{\mathcal{H}_{N}}e^{-\frac{N}{t}\tr Q(H) \pm \tr H^{q} \tr H^{q}} dH	
\end{equation*}
where $q$ is odd and V(H) is an even polynomial. This model is invariant under the transformation $H \mapsto -H$, hence its odd moments are zero. This further implies that its limiting odd moments are also zero.

For this model
\begin{equation*}
	V'(x) = Q'(x) + q m_{q} x^{q-1},
\end{equation*}
where $m_{q}$ denotes the $q$-th limiting moment. But this moment is zero so $V'(x)$ is the same for the model whether or not the multi-trace term is present. Hence, $W_{1}^{0}$ and therefore the limiting eigenvalue distribution are unaffected by odd multi-trace terms of this form.

\section*{Appendix B: Formal matrix integrals}

A matrix integral whether convergent or divergent can always be expanded in a perturbative series in terms of sums over maps \cite{Eynard2018,AK,blobbed}. In the case that the matrix integral is convergent the perturbative series is not necessarily the taylor expansion of the matrix integral.

 Consider for example the following quartic matrix  integral
\begin{equation*}
\int_{\mathcal{H}_{N}} e^{-N\left(\frac{1}{2}\tr H^{2} - \frac{t_{4}}{4}\tr H^{4} \right)}dH
\end{equation*}
Where $dH$ is the Lebesgue measure on the space of Hermitian $N$ by $N$ matrices $\mathcal{H}_{N}$, and $t_{4}$ is a coupling constant. For $t_{4}\leq 0$,  this integral is convergent and can be computed using orthogonal polynomials.  However, we are interested in, the not unrelated, formal summation
\begin{equation*}
Z = \sum_{k=0}^{\infty}	\int_{\mathcal{H}_{N}}\frac{N^{k}}{4^{k} k!}\,t_{4}^{k}\,(\tr H^{4})^{k} e^{-\frac{N}{2}\tr H^{2} }dH = \sum_{k=0}^{\infty}	\frac{N^{k}}{4^{k} k!}\,t_{4}^{k}\,\langle (\tr H^{4})^{k}\rangle_{0}, 
\end{equation*}
where the subscript zero denotes the expectation value with respect to the Gaussian random matrix integral above. We are also interested in its moments, which are themselves formal sums 
\begin{equation*}
\langle \tr H^{\ell}\rangle := \sum_{k=0}^{\infty}	\frac{N^{k}}{4^{k} k!}\,t_{4}^{k}\,\langle (\tr H^{\ell} \tr H^{4})^{k}\rangle_{0}. 
\end{equation*} 

Such formal integrals are well studied and have deep connections to areas of combinatorics, physics, and geometry \cite{Eynard2018}. For example these formal sums like those seen above have a realization as sums over maps. Furthermore, it is often the case that formal matrix models and their convergent counterparts (if they exist) coincide in the large $N$ limit. For more details on formal and convergent matrix models see \cite{formal models}.

\end{document}